\documentclass[a4paper,12pt]{article}
\usepackage{amsmath,amsthm,amssymb,setspace}
\usepackage{newtxtext,newtxmath}

\usepackage{hyperref}

\usepackage[margin=1.00in]{geometry}
\usepackage{subfigure}
\usepackage{bbm}
\usepackage{array,arydshln}
\usepackage{comment}
\usepackage[round]{natbib}

\let\sum\relax

\DeclareSymbolFont{CMlargesymbols}{OMX}{cmex}{m}{n}
\DeclareMathSymbol{\sum}{\mathop}{CMlargesymbols}{"50}

\newtheorem{theorem}{Theorem}

\newtheorem{lem}{Lemma}
\newtheorem{cor}{Corollary}
\theoremstyle{definition}

\newtheorem{as}{Assumption}

\begin{document}
	
	\title{
		Welfare ordering of voting weight allocations\footnote{I would like to thank Michel Le Breton, Yukio Koriyama, Junichiro Wada, and seminar participants at Osaka University and the 2022 Spring Meeting of Japanese Economic Association for their helpful comments and suggestions. Financial support from JSPS KAKENHI Grant (JP17K13706) and Waseda University Grant for Special Research Projects (2018K-014) is gratefully acknowledged. This paper subsumes part of the results in my earlier working paper, ``The likelihood of majority inversion in an indirect voting system,'' available at \url{https://ssrn.com/abstract=2785971}. Declarations of interest: none.}
	}
	
	\author{Kazuya Kikuchi\thanks{
			{Tokyo University of Foreign Studies. E-mail: \texttt{kazuya.kikuchi68@gmail.com}}.}}
    \date{This version: November, 2022}

	\maketitle
	\abstract{
		This paper studies the allocation of voting weights in a committee representing groups of different sizes. We introduce a partial ordering of weight allocations based on stochastic comparison of social welfare. We show that when the number of groups is sufficiently large, this ordering asymptotically coincides with the total ordering induced by the cosine proportionality between the weights and the group sizes. A corollary is that a class of expectation-form objective functions, including expected welfare, the mean majority deficit and the probability of inversions, are asymptotically monotone in the cosine proportionality.
		
	}

	\section{Introduction}
	
	How is social welfare affected by the voting weights in a committee representing groups of different sizes? Studies in the literature have identified cases in which the weights proportional to the populations of the groups are optimal (\citealt{BarberaJackson2006} and \citealt{BeisbartBovens2007}). However, those results do not tell us, for arbitrarily given two weight allocations, which one leads to higher welfare; whether increased proportionality implies welfare improvement; and, if it does, for what measure of proportionality.
	This paper extends the previous analysis by introducing a welfare ordering of weight allocations which is the intersection of various welfare indices employed in the literature, and relating it to a specific measure of weight-to-population proportionality.

	Our analysis is based on a standard model of two-stage voting in which the preferences of individuals are  aggregated first within groups such as states or countries, and then across the groups through weighted majority voting by group representatives. We evaluate the weight allocation among the groups, at a prior stage where the individual preferences are random variables. Social welfare is measured by the (random) number of individuals who prefer the voting outcome. We define the \textit{stochastic welfare ordering} to be the partial ordering of weight allocations that ranks allocation $a$ higher than allocation $b$ if the distribution of welfare under $a$ stochastically dominates that under $b$. This ordering is the intersection of the orderings induced by expectation-form objective functions, including expected welfare, the mean majority deficit, and the probability of inversions (see the discussion of the literature below).

	The main result is that under certain assumptions about the distribution of preferences, when the number of groups is sufficiently large, the stochastic welfare ordering asymptotically coincides with the total ordering induced by \textit{cosine proportionality}, a simple index measuring the proportionality between the weights and the populations (Theorem \ref{thm:main}).\footnote{\cite{koppel2009measuring} provide an axiomatic analysis of cosine proportionality.} It also follows, as a corollary, that the class of expectation-form social objectives mentioned above are asymptotically monotone in cosine proportionality, and that they thus asymptotically induce the same total ordering of weight allocations (Corollaries \ref{cor:u} and \ref{cor:mean}).

	\bigskip
	\noindent
	{\bf Relation to the literature.}
	\cite{BarberaJackson2006} study the design of two-stage voting and derive the optimal weights that maximize the total expected utility of individuals. Their result implies that under our assumptions, the optimal weights are proportional to the populations.\footnote{\citeauthor{BarberaJackson2006}'s \citeyearpar{BarberaJackson2006} result also implies that under an alternative assumption on the preference distribution, the optimal weights are proportional to the square roots of the populations. See Section \ref{sec:indep}.}

	\cite{felsenthal1999minimizing} introduce the quantity called the majority deficit, the gap between the size of the majority camp and the number of individuals preferring the voting outcome. They study the problem of minimizing the mean majority deficit. \cite{BeisbartBovens2008} use the mean majority deficit to compare alternative systems of the US presidential election. A closely related index is the probability of inversions, i.e., the probability that the majority of individuals disagree with the social decision. This probability is computed under various specific settings by, for example, \cite{May1948},
	\cite{Hinich1975}, \cite{feix2004probability},
	\cite{kaniovski2018probability},
	and
	\cite{de2020theoretical}).

	These papers concern either the optimal voting rule (including optimal weights), or the values of welfare indices under a specific weight allocation and/or population distribution.\footnote{The above review of the literature is by no means exhaustive. In particular, there are papers in the literature that consider social objectives not listed above. For instance, \cite{Penrose1946}, \cite{BeisbartBovens2007} and \cite{Kurz2017} study the optimal weights in terms of equity among individuals. \cite{Koriyamaetal2013} study optimal weights in terms of the total utility when the utility of each individual is a concave function of the frequency with which the social decision matches his preference.} Our focus is on deriving an approximate expression of a welfare ordering over the entire set of weight allocations, which provides insight into the problem of how welfare under suboptimal weight allocations can be compared. Our result implies that the above three examples of objective functions asymptotically induce the same ordering of weight allocations, which is represented by cosine proportionality. Corollary \ref{cor:mean} also shows the asymptotic formulas of these functions.

	\section{Model}
	\label{sec:model}

	{\bf Social decision.}
	Consider a society comprising $n$ groups labeled $i=1,\cdots,n$.
	The population of group $i$ is $s_i\in[0,\bar{s}]$.
	The group is assigned a voting weight $a_i=a(s_i)\in[0,\bar{a}]$ as a function of population $s_i$. We call function $a:[0,\bar{s}]\to[0,\bar{a}]$ the \textit{weight allocation}. We assume that $a$ is Borel-measurable and not almost everywhere zero.

	The society makes a choice between two alternatives, denoted $+1$ and $-1$, through two voting stages. First, in each group, all members cast one vote for their preferred alternative.
	Then the weight of each group is divided between the alternatives in a ratio that depends on the groupwide voting result.
	The society chooses the alternative that receives most weights in total.
	
	Formally, let $X_i\in[-1,1]$ be the group-$i$ vote margin between alternatives $+1$ and $-1$ in the first voting stage (i.e., the share of group-$i$ votes for $+1$ minus that for $-1$). 
	Let $r(X_i)\in[-1,1]$
	be the weight margin in the second stage (i.e., the share of group $i$'s weight for $+1$ minus that for $-1$). The function $r:[-1,1]\to[-1,1]$ represents the rule that converts the groupwide voting result into the division of the weight between the alternatives. We assume that $r$ is nondecreasing, odd and not identically zero. Examples of $r$ include winner-take-all $r(x)=\text{ sign }x$ and proportional representation $r(x)=x$.

	The social decision is $+1$ or $-1$, depending on the sign of the total weight margin $T=\sum a_ir(X_i)$:
	\[
	D=\begin{cases}
		{\rm sign\, } T &\text{ if }
		T\neq0\\
		\text{$\pm1$ equally likely}
		&\text{ if } T=0,
	\end{cases}
	\]
	where we assume that a coin is flipped in case of a tie.

	\bigskip
	\noindent
	{\bf Random preferences.}
	We consider the prior stage where individual preferences are random variables. This implies that the vote margins $X_i$ are random variables. We assume that the distribution of $X_i$'s satisfies the following:
	
	\begin{as}
		$X_1,X_2,\cdots$ are independent and identically distributed random variables, whose common distribution is nondegenerate and symmetric about 0.
		\label{as:x}
	\end{as}
	
	This assumption captures a situation in which individual preferences are correlated within the groups, independent across the groups, and ex ante neutral with respect to the alternatives.\footnote{The intragroup-correlation part of this interpretation is valid if we consider the asymptotic situation where the population of each group is sufficiently large and redefine $X_i$ to be the (almost sure) limit of the groupwide vote margin. If individual preferences are independent within group $i$, then the law of large numbers and the symmetry assumption imply that $X_i$ is degenerate at 0, violating the nondegeneracy assumption. Alternatively, if the preferences of members of group $i$ are correlated via random variable $X_i$ so that conditional on $X_i$, each member prefers $+1$ with probability $(1+X_i)/2$ and prefers $-1$ with probability $(1-X_i)/2$, then the law of large numbers implies that $X_i$ is indeed the limit of the groupwide vote margin.} It is well known that under this assumption, the perfectly proportional weight allocation (i.e., $a(s)\propto s$) maximizes the expected welfare. The object of this paper is to see how this observation extends to welfare comparison of arbitrary weight allocations.
	See Section \ref{sec:ext} for a discussion of the distributional assumption.

	\bigskip
	\noindent
	{\bf Welfare and proportionality.}
	We measure social welfare by the (random) number of individuals who prefer the social decision minus those who do not.\footnote{Section \ref{sec:utility} discusses the alternative definition of social welfare as the sum of utilities when preference intensities may vary across individuals.} This equals $DS$, where $S=\sum s_iX_i$ is the total vote margin between alternatives $+1$ and $-1$. For convenience, we normalize it by dividing by the standard deviation $\sigma=\sqrt{\mathbf{E}(X_1^2)\times\sum s_i^2}$:
	\[
	W=\frac{DS}{\sigma}.
	\]
	We denote by $H_a$
	the distribution of welfare $W$, which depends on the weight allocation $a$.
	
	The \textit{stochastic welfare ordering} is the partial ordering $\succsim$ on the set of weight allocations,
	defined as follows: for any two weight allocations $a$
	and $b$,
	$a\succsim b$
	if $H_a$ stochastically dominates $H_b$.\footnote{We say that a distribution $F$ \textit{stochastically dominates} another distribution $G$ if $F(x)\leq G(x)$ for all $x$.}

	Apart from the welfare properties, weight allocations can also be evaluated in terms of the proportionality to the populations. There are various measures of proportionality. One such measure is \textit{cosine proportionality}, defined by the cosine of the angle between the vector of populations in the groups and the vector of weights of the groups:
	\[
	c(a)=\frac{\sum s_ia_i}{\sqrt{
			\sum s_i^2\sum a_i^2
	}}.
	\]
	
	\bigskip
	\noindent
	{\bf Large number of groups.}
	To obtain a clear analytical result, we focus on the asymptotic situation in which the number of groups is sufficiently large.
	We introduce the sequence of societies, indexed by $n=1,2,3\cdots$, in which the $n$th society consists of the $n$ groups $i=1,\cdots,n$.

	We claim that
	\textit{as $n\to\infty$,
		the welfare distribution
		$H_a$ converges in law
		to a nondegenerate
		distribution
		$H_{a}^\ast$}.
	This claim is proved in the next section (Lemma \ref{lem:sn}).
	For now, let us take the existence of an asymptotic distribution as true, and define the \textit{asymptotic stochastic welfare order} $\succsim^\ast$ as follows:
	for any two weight allocations $a$ and $b$, $a\succsim^\ast b$ if $H_a^\ast$ stochastically dominates $H_b^\ast$.

	To define the asymptotic version of cosine proportionality, we need the following assumption on the sequence of group sizes.

	\begin{as}
		Let $\Psi^n$ be the distribution of group sizes in the $n$th society:
		$\Psi^n(s)=\sharp\{i\leq n|s_i\leq s\}/n \text{ for }
		s\in[0,\bar{s}]$.
		As $n\to\infty$, $\Psi^n$ weakly converges to a distribution $\Psi^\ast$, not degenerate at 0.
	\end{as}
	
	Under this assumption, the 
	\textit{asymptotic
		cosine proportionality} of weight allocation $a$
	is defined as the limit of $c(a)$ as $n\to\infty$, which is
	\[
	c^\ast(a)=
	\frac{\int sa(s)d\Psi^\ast(s)}{
		\sqrt{
			\int
			s^2 d\Psi^\ast(s)
			\int a(s)^2
			d\Psi^\ast(s)
		}
	}.
	\]

	\section{Asymptotic equivalence theorem}
	\label{sec:result}

	\begin{theorem}
		The asymptotic stochastic welfare order coincides with the total order induced by asymptotic cosine proportionality. That is, for any weight allocations $a$ and $b$,
		\[
		a
		\succsim^\ast
		b
		\Leftrightarrow\,
		c^\ast(a)\geq
		c^\ast(b).
		\]
		\label{thm:main}	
	\end{theorem}

	We divide the proof of the theorem
	into the
	two lemmas below.
	We need some notations.
	Let $\tau=\sqrt{\mathbf{E}\left[r(X_1)^2\right]\sum
		a_i^2}$ be the standard deviation of the total weight margin $T=\sum a_ir(X_i)$. Recall also that $S=\sum s_iX_i$ denotes the total vote margin, and $\sigma=\sqrt{\mathbf{E}\left[X_1^2\right]\sum
		s_i^2}$ its standard deviation.
	Let $\rho$
	be the correlation coefficient
	of
	$X_{i}$
	and $r(X_{i})$:
	\begin{equation*}
		\begin{split}
			\rho=\text{Corr}(X_i,r(X_i))&=
			\frac{\mathbf{E}\left[X_ir(X_i)\right]}{\sqrt{
					\mathbf{E}\left[
					X^2_i
					\right]
					\mathbf{E}\left[
					r(X_i)^2
					\right]
			}}>0.
		\end{split}
	\end{equation*}
	The positiveness
	of $\rho$
	can be checked
	by noting that function $r$
	is an odd and nondecreasing function
	not identically zero,
	and that $X_i$ has a symmetric and
	nondegenerate distribution.
	
	\begin{lem}
		As $n\to\infty$,
		the
		random vector
		$\left(
		S/\sigma,T/\tau
		\right)$
		converges in law
		to the bivariate
		normal distribution
		with zero means,
		unit variances
		and correlation
		coefficient
		$\rho c^\ast(a)$.
		\label{lem:clt}
	\end{lem}
	
	\begin{proof}[Proof of Lemma \ref{lem:clt}]
		We use the superscript
		$n$ to clarify
		the number of groups.
		Fix two real numbers
		$\xi$
		and $\eta$
		not both zero, and
		let
		\[
		Z^n=
		\xi
		\frac{S^n}{\sigma^n}+
		\eta
		\frac{T^n}{\tau^n}.
		\]
		The variance
		of $Z^n$
		is 
		\[
		(\zeta^n)^2=\xi^2+\eta^2+2\xi\eta
		\rho c^n(a),
		\]
		where $c^n(a)$
		denotes cosine proportionality.

		By the Cram\'{e}r-Wold
		device,
		the proof of the lemma
		reduces to
		showing
		that
		$S^n$
		converges in law
		to the normal distribution
		with mean 0
		and variance
		$\xi^2+\eta^2+2\xi\eta
		\rho c^\ast(a)$.
		Since $c(a)\to c^\ast(a)$, we have
		\begin{equation}
			(\zeta^n)^2\to
			\xi^2+\eta^2+2\xi\eta
			\rho c^\ast(a).
			\label{eq:zeta}
		\end{equation}
		Thus it suffices to show that
		$Z^n/\zeta^n$
		converges in law to
		the standard normal distribution
		$N(0,1)$.

		Note that $Z^n$
		is a sum of
		independent
		random variables:
		\[
		Z^n=\sum_{i=1}^nR_{i}^n,
		\text{ where } 
		R_{i}^n=
		\xi\frac{s_iX_i}{\sigma^n}
		+\eta
		\frac{a(s_i)r(X_i)}{\tau^n}.
		\]
		To show that
		$Z^n/\zeta^n$
		converges in law
		to $N(0,1)$,
		we apply the Lindeberg Central
		Limit Theorem  to the triangular array $\{R_{i}^n:\,i\leq 
		n;\,
		n=1,2,\cdots\}$.\footnote{
			See, e.g., Theorem 27.2
			in \cite{billingsley2008probability}.}
		We need to check Lindeberg's condition:
		\begin{equation}
			\text{For
				any $\delta>0$,
			}\sum_{i=1}^n\mathbf{E}\left[
			\left(\frac{R_{i}^n}{
				\zeta^n}
			\right)^2
			\mathbf{1}_{\left\{
				\left(\frac{R_{i}^n}{
					\zeta^n}
				\right)^2
				>\delta\right\}}
			\right]\to0.
			\label{eq:lindeberg}
		\end{equation}
		In the definition
		of $R_i^n$, the numerators 
		$s_{i}X_{i}$ and $a(s_{i})r(X_{i})$
		are bounded,
		while the denominators $\sigma^n$
		and $\tau^n$ tend
		to infinity ($\zeta^n$
		converges to the finite limit
		(\ref{eq:zeta})).
		Thus, for any $\delta>0$,
		there exists $n_0$
		such that
		for all $n>n_0$
		and all $i=1,\cdots,n$,
		$\left(R_{i}^n/
		\zeta^n
		\right)^2
		<\delta$
		almost surely,
		which implies that
		the sum
		in (\ref{eq:lindeberg})
		is zero
		for all $n>n_0$.
	\end{proof}

	The \textit{skew normal
		distribution}\footnote{
		See
		\cite{azzalini1996multivariate} for basic
		properties
		of the skew normal
		distribution.}
	with parameter $\lambda$,
	denoted
	$SN(\lambda)$,
	is the distribution
	with the density function
	\[
	2\phi(x)\Phi(\lambda x)\text{ for }
	x\in\mathbf{R},
	\]
	where
	$\phi$
	and $\Phi$
	denote the
	density function
	and the distribution
	function
	of the standard
	normal distribution
	$N(0,1)$.

	\begin{lem}
		As $n\to\infty$, social welfare
		$W\to W^\ast$
		in law, where $W^\ast$ has 
		the skew normal
		distribution
		$SN(\lambda_a)$
		with parameter
		$\lambda_a=
		\rho c^\ast(a)/\sqrt{1-\rho^2
			c^\ast(a)^2}$.
		The distribution $SN(\lambda_a)$ is stochastically
		increasing in the asymptotic
		cosine proportionality $c^\ast(a)$.
		\label{lem:sn}
	\end{lem}

	\begin{proof}[Proof of Lemma
		\ref{lem:sn}]
		Welfare $W$ can be written as
		\[
		W=
		\left(
		\text{sign }
		\frac{T}{\tau}+L
		\mathbf{1}_{
			\left\{
			\frac{T}{\tau}
			=0
			\right\}}
		\right)\times
		\frac{S}{\sigma},
		\]
		where $L$ is a tie-breaking variable that takes values $\pm1$ with equal probabilities and is independent of all other variables.
		By Lemma \ref{lem:clt},
		$\left(
		S/\sigma,
		T/\tau
		\right)$
		converges in law to a random vector
		$(S^\ast,T^\ast)$
		that has
		the bivariate normal distribution with zero means,
		unit variances
		and correlation coefficient
		$\rho c^\ast(a)$.
		Since $T^\ast\neq0$
		almost surely
		and the function
		$f(s,t)=s\text{\,sign\,}t$
		is continuous
		except on a set of probability zero,
		$W$
		converges in law
		to the random variable
		$W^\ast=S^\ast
		\text{ sign }T^\ast$.
		
		We now check that
		the distribution of $W^\ast$
		is $SN(\lambda_a)$ as stated
		in the lemma, by
		conditioning on
		the sign of $T^\ast$.
		The conditional distribution
		of
		$W^\ast$
		given $T^\ast>0$
		equals
		the conditional distribution
		of $S^\ast$
		given $T^\ast>0$.
		By Proposition 2 in \cite{azzalini1996multivariate},
		the latter conditional distribution
		is the 
		skew normal distribution
		$SN(\lambda_a)$
		with the parameter $\lambda_a$
		specified in the lemma.
		Similarly,
		the conditional distribution
		of $W^\ast$
		given $T^\ast<0$ is also $SN(\lambda_a)$.

		For the last
		sentence
		of the lemma,
		it suffices to show that
		the skew normal distribution
		$SN(\lambda)$
		is stochastically increasing
		in $\lambda$.
		Let $H(x;\lambda)$
		be the distribution function
		of $SN(\lambda)$:
		\[
		H(x;\lambda)=\int_{-\infty}^x
		2\phi(y)\Phi(\lambda y)dy.
		\]
		We show
		that $\frac{\partial H}{\partial \lambda}<0$.
		The derivative
		is
		$\frac{\partial H}{\partial \lambda}
		=\int_{-\infty}^x
		2y\phi(y)\phi(\lambda y)dy$.
		This expression is proportional
		to the conditional expectation $\mathbf{E}(Y|Y\leq x)$,
		where $Y$
		is a random variable having a
		density function proportional to
		$\phi(y)\phi(\lambda
		y)$.
		Obviously $Y$
		is normally distributed with mean 0.
		Thus
		$\int_{-\infty}^x
		2y\phi(y)\phi(\lambda y)dy
		\propto\mathbf{E}(Y
		|Y\leq x)<\mathbf{E}(Y)=0$
		and therefore $\frac{\partial H}{\partial \lambda}<0$.
	\end{proof}

	\section{Expectation-form objective functions}
	\label{sec:vnm}
	
	We discuss
	the implications
	of the asymptotic equivalence theorem
	to social objective functions of the form $\mathbf{E}[f(W)]$ for a nondecreasing function $f$.

	To deal with the limits of expectations, we restrict the class of functions $f$. We say that a function $f:\mathbf{R}\to\mathbf{R}$
	has a \textit{square exponential bound}	if
	there exist constants
	$\alpha>0$ and $\beta\in(0,1)$
	such that
	\[
	|f(w)|\leq \exp\left(\alpha+\beta w^2
	\right)
	\]
	for all $w\in\mathbf{R}$.
	This condition is fairly weak.
	It allows
	$f(w)$
	to increase
	as fast as
	$e^{\beta w^2}$
	when $w\to\infty$
	and decrease
	as fast as
	$-e^{\beta w^2}$
	when $w\to-\infty$,
	for some $\beta\in(0,1)$.

	\begin{cor}
		Let
		$f$ be any nondecreasing function
		with a square exponential bound.
		\begin{itemize}
			\item[(i)]
			As $n\to\infty$,
			$\mathbf{E}\left[
			f(W)\right]\to\mathbf{E}\left[
			f(W^\ast)\right]$, where $W^\ast$ has the skew normal distribution $SN(\lambda_a)$ with parameter $\lambda_a=
			\rho c^\ast(a)/\sqrt{1-\rho^2
				c^\ast(a)^2}$.
			\item[(ii)]
			The limit
			$\mathbf{E}\left[
			f(W^\ast)\right]$ is a nondecreasing function of the asymptotic cosine proportionality $c^\ast(a)$.
		\end{itemize}
		\label{cor:u}
	\end{cor}

	\begin{proof}
		(i) We use superscript $n$ to indicate the number of groups. By Lemma \ref{lem:sn}, $W^n\to W^\ast$ in law.
		Since $f$ is nondecreasing,
		the set of its discontinuity
		points $\Delta$
		has Lebesgue measure 0.
		Since $W^\ast$
		is continuously distributed,
		$\mathbf{P}\{W^\ast\in \Delta\}=0$.
		Thus,
		by Slutsky's theorem,\footnote{
			See e.g. \cite{ferguson2017course}.	
		}
		$f(W^n)\to f(W^\ast)$
		in law.

		If, in addition,
		the sequence of random variables
		$\{f(W^n)\}$
		is uniformly integrable (UI),
		then statement (i) follows.\footnote{
			A sequence
			$\{X^n\}$
			of random 
			variables
			is called
			\textit{uniformly
				integrable}
			if for any
			$\epsilon>0$,
			there exists
			$x>0$
			such that
			$\mathbf{E}
			(|X^n|\mathbf{1}_{\{|X^n|>x\}})
			<\epsilon$
			for all $n$.
			See e.g.
			\cite{williams1991probability}.}
		To show that
		$\{f(W^n)\}$
		is UI,
		it suffices
		to show that the dominating
		sequence
		$\{Y^n\}:=
		\{\exp(\alpha+\beta(W^n)^2)\}$
		is UI.
		Without loss of generality,
		assume $\alpha=0$.
		A sufficient
		condition for uniform integrability
		is boundedness in $L^p$
		for some $p>1$,
		i.e.,
		that there is a constant $K<\infty$
		with
		$\mathbf{E}\left[
		(Y^n)^p
		\right]<K$ for all $n$.

		To check the above
		sufficient
		condition,
		choose $p>1$
		so that $p\beta\in(0,1)$.
		Note that
		\[
		\mathbf{E}\left[
		(Y^n)^p
		\right]
		=\int_0^\infty
		\mathbf{P}\left\{
		(Y^n)^p>y
		\right\}dy
		\leq 
		1+
		\int_1^\infty
		\mathbf{P}\left\{
		(Y^n)^p>y
		\right\}dy.
		\]
		Rewrite
		$\mathbf{P}\left\{
		(Y^n)^p>y
		\right\}
		=
		\mathbf{P}\left\{
		(W^n)^2>
		\frac{\log y}{p\beta}
		\right\}$ for $y>1$.
		By definition,
		$(W^n)^2=
		\left(\sum_{i=1}^n
		s_iX_i/\sigma
		\right)^2$,
		the square
		of a sum of independent
		variables with mean 0,
		where the sum
		is normalized to unit variance.
		We can thus apply Hoeffding's inequality
		to get
		\begin{equation*}
			\begin{split}
				\mathbf{P}\left\{
				(W^n)^2>\frac{\log y}{p\beta}
				\right\}
				&\leq
				2y^{-\frac{1}{p\beta}}.
			\end{split}
		\end{equation*}
		Thus, recalling
		that $p\beta\in(0,1)$,
		\begin{equation*}
			\begin{split}
				\int_1^\infty
				\mathbf{P}\left\{
				(Y^n)^p>y
				\right\}dy
				&\leq
				\int_1^\infty
				2y^{-\frac{1}{p\beta}}
				dy
				=\frac{2p\beta}{1-p\beta}.
			\end{split}
		\end{equation*}
		Therefore
		$\mathbf{E}\left[
		(Y^n)^p\right]
		\leq (1+p\beta)/(1-p\beta)
		<\infty$ for all $n$.

		\bigskip
		\noindent
		(ii)
		This follows directly from the second sentence in Lemma \ref{lem:sn}.
	\end{proof}

	\bigskip
	\noindent
	{\bf Examples.}
	We provide three examples of expectation-form objective functions.
	The simplest one is
	the \textit{expected welfare}:
	\[
	u(a)=\mathbf{E}[W],
	\]
	which corresponds to the linear function
	$f(w)=w$.

	A related index is
	the \textit{mean majority
		deficit}, which
	was introduced by
	\cite{felsenthal1999minimizing}.
	It is the expected difference
	between
	the size of
	the majority camp,
	$\left(|S|+\sum s_i
	\right)/2$,
	and the
	number
	of individuals
	who prefer the social decision,
	$\left(DS+\sum s_i
	\right)/2$; that is,
	$\mathbf{E}\left(|S|-
	DS
	\right)/2$.
	We redefine it by dividing by $\sigma$:
	\[
	\delta(a)=\frac{1}{2}\mathbf{E}\left(
	\frac{|S|}{\sigma}-
	W
	\right).
	\]
	Note that the distribution
	of $|S|/\sigma$ does not depend
	on the weight allocation $a$.
	Thus 
	the negative of
	the mean majority deficit
	$-\delta(a)$
	can be represented as the expectation
	$\mathbf{E}\left[
	f(W)
	\right]$
	for the positive affine function
	$f(w)=
	\left(w-\mathbf{E}(|S|)/\sigma\right)/2$.

	Another
	index
	that has been
	extensively studied
	in the literature
	is
	the \textit{probability of inversion}
	(or the
	\textit{probability of the referendum
		paradox}).
	It is the probability
	that the majority of individuals
	dislike the social decision:
	\[
	p(a)=
	\mathbf{P}\{W<0\}.
	\]
	The complementary probability $1-p(a)$ is the expectation
	$\mathbf{E}\left[
	f(W)
	\right]$
	for the
	nondecreasing step function
	$f=\mathbf{1}_{\{w\geq0\}}$.

	For finite $n$,
	these three indices 
	induce different complete extensions
	of the stochastic welfare ordering $\succsim$ over weight allocations. Since the functions $f$ corresponding to these indices have square exponential bounds,
	Corollary \ref{cor:u}
	implies that they are asymptotically monotone in cosine proportionality.
	In fact, Lemma \ref{lem:sn}
	allows us to
	derive the limiting
	monotone
	functions explicitly.

	\begin{cor}
		As $n\to\infty$,
		\begin{enumerate}
			\item[(i)]
			$u(a)\to
			\sqrt{\frac{2}{\pi}}\rho
			c^\ast(a)$.
			\item[(ii)]
			$\delta(a)
			\to
			\frac{1-\rho c^\ast(a)}{\sqrt{2\pi}}$.
			\item[(iii)]
			$p(a)\to
			\frac{\arccos\left(\rho
				c^\ast(a)\right)}{\pi}$.\footnote{\cite{Hinich1975}
				derives a similar limit formula
				of $p(a)$, assuming
				equally sized groups.
			}
		\end{enumerate}
		\label{cor:mean}
	\end{cor}

	\begin{proof}
		(i) The mean of the skew normal distribution
		$SN(\lambda)$
		is 
		$\sqrt{2/\pi}\cdot
		\lambda/\sqrt{1+\lambda^2}$.\footnote{
			See, e.g., \cite{azzalini1996multivariate}.}
		Plugging $\lambda_a$ from Lemma \ref{lem:sn} yields (i).
		
		\noindent
		(ii)
		In the definition of $\delta(a)$, the variable $|S|/\sigma$
		converges in law to the standard half-normal distribution,
		whose mean is $\sqrt{2/\pi}$.
		Combining this with (i) proves (ii).

		\noindent
		(iii)
		The limit probability of inversion
		is 
		\[
		\mathbf{P}\{W^\ast<0\}=\int_{-\infty}^02\phi(x)\Phi(\lambda_a
		x)dx=
		\frac{1}{\pi}\arctan
		\frac{1}{\lambda_a}.\footnote{
			The second equality
			follows, e.g., from Formula 1,010.3 in
			\cite{owen1980table}.}
		\]
		Noting that $1/\lambda_a=
		\tan\theta$
		for $\theta\in[0, \pi/2]$
		satisfying $\cos\theta=
		\rho c^\ast(a)$ proves (iii).
	\end{proof}

	\section{Extensions}
	\label{sec:ext}
	
	\subsection{Preference intensities}
	\label{sec:utility}
	
 We have shown Theorem \ref{thm:main} based on the definition of social welfare as the (normalized) number of individuals who prefer the social decision, which ignores cardinal differences between individual preferences. This section discusses an extension of the model for which the theorem is valid, even if heterogeneous preference intensities are incorporated in the definition of social welfare.

	Let $m_i$ denote the population of group $i$ and label its members by $k=1,\cdots,m_i$. 
	We denote by $U_{ik}$ the utility member $k$ of group $i$ obtains from alternative $+1$, and fix the utility from alternative $-1$ to be 0. The member thus votes for $+1$ or $-1$, depending on whether $U_{ik}$ is positive or negative. We assume that $U_{ik}$ has the following form:
	\[
	U_{ik}=\Theta_i+\epsilon_{ik}.
	\]
	The term $\Theta_i$ summarizes group-specific factors affecting the utility, while $\epsilon_{ik}$ summarizes individual-specific factors.
	We make the following assumptions on the distribution of these variables:
	(i) All $\Theta_i$ and $\epsilon_{ik}$ ($i=1,2,\cdots$; $k=1,2,\cdots$) are independent. (ii) All $\Theta_i$ have a common symmetric and nondegenerate distribution $G_\Theta$ on $[-\bar{\theta},\bar{\theta}]$. (iii) All $\epsilon_{ik}$ have a common symmetric and nondegenerate distribution $G_\epsilon$ on $[-\bar{\epsilon},\bar{\epsilon}]$.

	Consider the large-population limit as the absolute size of each group $m_i$ tends to infinity, and the relative size of group $i$ normalized by group 1's size, $m_i/m_1$, tends to the constant $s_i\in[0,\bar{s}]$.	
	Then, \textit{almost surely, the group-$i$ average utility of alternative $+1$, $\sum_{k=1}^{m_i}U_{ik}/m_i$, converges to $\Theta_i$, while the group-$i$ vote margin in favor of alternative $+1$, $\sum_{k=1}^{m_i}\text{\rm sign }U_{ik}/m_i$, converges to $2G_\epsilon(\Theta_i)-1$.} The first part of this claim follows from a direct application of the law of large numbers. The second part can be proved as follows: by the law of large numbers, $\sum_{k=1}^{m_i}\text{\rm sign }U_{ik}/m_i\to \mathbf{E}[\text{\rm sign }U_{ik}|\Theta_i]=\mathbf{P}\{U_{i1}>0|\Theta_i\}-\mathbf{P}\{U_{i1}<0|\Theta_i\}=2G_\epsilon(\Theta_i)-1$ almost surely.
	
	The total utility from alternative $+1$ is $\sum s_i \{2G_{\epsilon}(\Theta_i)-1\}$, and hence the total utility from the social decision is $D\times\sum s_i \{2G_{\epsilon}(\Theta_i)-1\}$. \textit{If we redefine social welfare $W$ to be (a normalization of) the total utility, then Theorem \ref{thm:main} holds.} The proof only requires us to redefine $X_i:=\Theta_i$, $S:=\sum s_i\{2G_{\epsilon}(\Theta_i)-1\}$,
	and $T:=\sum a(s_i)r(\Theta_i)$;
	then Lemmas \ref{lem:clt} and \ref{lem:sn} hold with correlation coefficient $\rho:=\text{\rm Corr }[2G_{\epsilon}(\Theta_i)-1,r(\Theta_i)]$, and the theorem follows.

	\subsection{Independent preferences}
	\label{sec:indep}

		Assumption \ref{as:x} captures a situation
		in which individual preferences 
		are correlated within the groups.
		This section discusses how the results in this paper change under the alternative assumption that the preferences are independent within the groups.

		We use the same basic model of preferences
		as in Section \ref{sec:utility}, but assume that $\Theta_i\equiv0$ and $\epsilon_{ik}=\pm1$ equally likely. Thus, each individual independently prefers alternative $+1$ or $-1$ with equal probabilities.
		Consider the large-population limit
		as $m_1\to\infty$
		and $m_i/m_1\to s_i$
		for each group $i$.
		The following lemma	is a direct consequence of the central limit theorem.
		
		\begin{lem}
			The group-$i$ vote margin
			in favor of alternative $+1$
			multiplied by $\sqrt{m_1}$, i.e.,
			$\sqrt{m_1}\sum_{k=1}^{m_i}\epsilon_{ik}/m_i$,
			converges
			in law to 
			$X_i:=\frac{N_i}{\sqrt{s_i}}$,
			where $N_i$
			($i=1,\cdots,n$)
			are independent
			and have the
			standard
			normal distribution $N(0,1)$.
			\label{lem:lim_pop_ic}
		\end{lem}
		
		Note the following two differences from the case with intragroup correlations of preferences:
		first, $X_i$ are 
		\textit{not} identically distributed, since the variances are smaller for larger groups;
		second, the support of $X_i$ is $\mathbf{R}$, not $[-1,1]$.

		We provide
		an analog
		of Theorem \ref{thm:main}
		for this independent-preference 
		model, but not in full
		generality.
		We only consider two specific representation rules:
		\textit{winner-take-all}
		$r(x)=\text{sign\,} x$
		and 
		\textit{proportional representation}
		$r(x)=x$.

		\begin{theorem}
			The following statements hold in the independent-preference model:
			\begin{itemize}
				\item[(i)]
				Under winner-take-all, as the number of groups $n\to\infty$,
				the asymptotic stochastic welfare ordering
				$\succsim^\ast$
				coincides with
				the ordering induced by the following index:
				\[
				c_{\rm sqrt}^\ast
				(a)=
				\frac{\int
					\sqrt{s}a(s) d\Psi^\ast(s)}{
					\sqrt{
						\int s d\Psi^\ast(s)
						\int a(s)^2 d\Psi(s)
						}
				}.
				\]
				\item[(ii)]
				Under proportional representation, for any number $n$ of groups,
				the stochastic welfare ordering
				$\succsim$
				coincides with
				the ordering induced by the following index:
				\[
				\hat{c}_{\rm sqrt}
				(a)=
				\frac{\sum a(s_i)
					}{
					\sqrt{
						\sum s_i 
						\sum \frac{a(s_i)^2}{s_i}
					}
				}.
				\]
			\end{itemize}
			\label{thm:ic}
		\end{theorem}

		\begin{proof}
			The total
		vote margin 
			and the
			total weight
			margin in favor of alternative $+1$ are
			\[
			S=
			\sum_{i=1}^n\sqrt{s_{i}}N_{i},\,
			T=
			\sum_{i=1}^n
			a(s_{i})r(X_i).
			\]
			The variance
			of $S$
			is $\sigma^2=
			\sum s_{i}$.
			We begin with the proof of
			part (ii).
			
			\bigskip
			\noindent
			\textit{Proof of (ii)}.
			Under proportional representation,
			$r(X_{i})=X_i=
			\frac{N_{i}}{\sqrt{s_{i}}}$.
			The variance of $T$
			is 
			$\tau^2=
			\sum_{i=1}^n\frac{a(s_{i})^2}{s_{i}}$.
			Since a sum
			of independent
			and normally distributed
			random variables
			is again normally distributed,
			it follows that
			for any $n$,
			$\left(
			\frac{S}{\sigma},
			\frac{T}{\tau}
			\right)$
			has the centered normal distribution
			with
			unit variances
			and correlation coefficient
			$\hat{c}_{\rm sqrt}(a)$.
			The same argument as in the
			proof of Lemma \ref{lem:sn}
			shows that
			welfare $W$
			has the skew normal distribution
			with parameter
			$\lambda=
			\frac{\hat{c}_{\rm sqrt}(a)}{\sqrt{
					1-\hat{c}_{\rm sqrt}(a)^2
			}}$,
			and hence that
			$W$
			is stochastically
			increasing in $\hat{c}_{\rm sqrt}(a)$.
			
			\bigskip
			\noindent
			\textit{Proof of (i)}.
			Under winner-take-all,
			$r(X_i)=\text{ sign }N_i$.
			The variance of $T$
			is now $\tau^2=\sum a(s_{i})^2$.
			We first
			show that
			$\left(
			\frac{S}{\sigma},
			\frac{T}{\tau}
			\right)$
			converges
			in law
			to the centered
			normal distribution
			with unit variances
			and correlation
			coefficient
			$\sqrt{2/\pi}c_{\rm sqrt}^\ast(a)$.
			As in the
			proof of
			Lemma \ref{lem:clt}, 
			for any
			two real numbers
			$\alpha$ and $\beta$
			not both zero,
			let
			\[
			Z^n=\alpha\frac{S^n}{\sigma^n}
			+\beta\frac{T^n}{\tau^n}=
			\sum_{i=1}^nR_{i}^n,\text{
				where
			}R_{i}^n=\alpha\frac{\sqrt{s_{i}}N_{i}}{\sigma^n}+
			\beta\frac{a(s_{i}){\rm\,sign\,}N_{i}}{\tau^n}.
			\]
			The variance
			of $Z^n$
			is 
			$(\zeta^n)^2=
			\alpha^2+\beta^2+
			2\sqrt{\frac{2}{\pi}}\alpha\beta
			c_{\rm sqrt}(a)$, where we define $c_{\rm sqrt}(a):=\sum\sqrt{s_i}a(s_i)/
			\sqrt{\sum s_i\sum a(s_i)^2}$.\footnote{The computation
				of the variance uses the fact that
				$\mathbf{E}(Z_{i} {\rm\,sign\,}Z_{i})=
				\mathbf{E}(|Z_{i}|)=
				\sqrt{\frac{2}{\pi}}$.}
			By the Cram\'{e}r-Wold
			Device,
			it suffices
			to show that
			$\frac{S^n}{\zeta^n}$
			converges in law
			to the standard
			normal distribution.
			To do this,
			we use the Lyapunov
			Central Limit Theorem
			(Theorem 27.3 in \citealt{billingsley2008probability}).
			We need to check Lyapunov's condition:
			\[
			\frac{1}{(\zeta^n)^4}\sum_{i=1}^n
			\mathbf{E}[(R_{i}^n)^4]\to0.
			\]
			Minkowski's
			inequality
			implies that
			for all $i=1,\cdots,n$,
			\[
			\left[\mathbf{E}\left(
			(R_{i}^n)^4
			\right)\right]^{\frac{1}{4}}
			\leq
			\frac{|\alpha|
				\sqrt{\bar{s}}}{
				\sigma^n}
			\left[\mathbf{E}(N_{1}^4)\right]^{\frac{1}{4}}
			+
			\frac{|\beta|\bar{a}}{\tau^n}
			=:K_n.
			\]
			Since
			both $\sigma^n$
			and $\tau^n$
			are $O(\sqrt{n})$,
			$K_n=O\left(1/\sqrt{n}\right)$.
			Noting that $\zeta^n$
			converges to a positive 
			finite limit, we have
			\[
			\frac{1}{(\zeta^n)^4}\sum_{i=1}^n
			\mathbf{E}[(R_{i}^n)^4]
			\leq
			\frac{nK_n^4}{(\zeta^n)^4}
			\to0.
			\]
			This completes
			the proof
			for asymptotic normality
			of $\left(
			\frac{S}{\sigma},
			\frac{T}{\tau}
			\right)$.
			The remaining
			step
			that the limit
			distribution of
			$W$
			is stochastically
			increasing
			in $c^\ast_{\rm sqrt}(a)$
			is the same
			as in the
			proof of Lemma
			\ref{lem:sn}.
		\end{proof}

		Theorem \ref{thm:ic}
		shows that
		in the independent-preference model,
		the (asymptotic) stochastic welfare ordering of weight allocations differs from the ordering induced by plain cosine proportionality, and depends on the representation rule $r$.
		Under winner-take-all, the stochastic welfare ordering is asymptotically equivalent to the cosine between the weights of the groups and the \textit{square roots} of the populations in the groups. This is consistent with what is known in the literature as the ``square root law,'' which states that in a model with independent preferences, under winner-take-all, the weights proportional to the square roots of populations are optimal in terms of social welfare, and also in terms of equity among individuals (e.g., \citealt{BarberaJackson2006}, \citealt{BeisbartBovens2007}, and \citealt{Penrose1946}).
		Under proportional representation, the equivalence theorem holds for any number of groups, but with a somewhat complicated index: the cosine proportionality between the square roots of the populations ($\sqrt{s_i}$) and the weights divided by the square roots of the populations ($a_i/\sqrt{s_i}$). Interestingly, if the total weight $\sum a_i$ is fixed, this index is ordinally equivalent to (the negative of) the well-known Saint-Lagu\"{e} index of disproportionality: $\sum a_i^2/s_i$ (see, e.g., \citealt{balinski2010fair}).

	\phantomsection
	\addcontentsline{toc}{section}{References}
	
	\bibliography{bibbib}

\end{document}